\documentclass[11pt]{article}
\usepackage{booktabs} % For formal tables
\usepackage{fullpage}
\usepackage{kpfonts}
\usepackage{color}
\usepackage{xcolor}
\usepackage{amsthm}

\usepackage[algo2e, noline, ruled]{algorithm2e}

\usepackage{varwidth}
\usepackage{algorithmic}
\SetKwRepeat{Do}{do}{while}%

\SetAlFnt{\small}
\SetAlCapFnt{\small}
\SetAlCapNameFnt{\small}
\SetAlCapHSkip{0pt}
\IncMargin{-\parindent}

\definecolor{DarkGreen}{rgb}{0.1,0.5,0.1}
\definecolor{DarkRed}{rgb}{0.5,0.1,0.1}
\definecolor{DarkBlue}{rgb}{0.1,0.1,0.5}
\usepackage[]{hyperref}
\hypersetup{
    unicode=false,          % non-Latin characters in Acrobat's bookmarks
    pdftoolbar=true,        % show Acrobat toolbar?
    pdfmenubar=true,        % show Acrobat menu?
    pdffitwindow=false,      % page fit to window when opened
    pdftitle={},    % title
    pdfauthor={}
    pdfsubject={},   % subject of the document
    pdfnewwindow=true,      % links in new window
    pdfkeywords={keywords}, % list of keywords
    colorlinks=true,       % false: boxed links; true: colored links
    linkcolor=DarkRed,          % color of internal links
    citecolor=DarkGreen,        % color of links to bibliography
    filecolor=DarkRed,      % color of file links
    urlcolor=DarkBlue,          % color of external links
}
\usepackage{xspace}

\usepackage{cleveref}
\usepackage{color-edits}
\addauthor{jm}{red}
\addauthor{mk}{green}
\addauthor{ar}{blue}
\addauthor{sw}{cyan}
\addauthor{mp}{red}
\addauthor{rv}{purple}
\addauthor{sk}{green}

\usepackage{hyperref}
\usepackage{amsopn}
\usepackage{natbib}

\DeclareMathOperator*{\argmax}{arg\,max}
\newcommand{\E}[1]{\mathbb{E}\lbrack #1 \rbrack}
\newcommand{\Ex}[2]{\mathbb{E}_{#1}\left\lbrack #2 \right\rbrack}

\renewcommand{\Pr}[2]{\mathbb{P}_{#1}\left\lbrack #2 \right\rbrack}

\newcommand{\vecdot}[2]{\langle {#1}, {#2} \rangle \xspace}

\newcommand{\chainfair}{\textsc{ChainedFair}\xspace}
\newcommand{\findchain}{\textsc{FindChained}\xspace}
\newcommand{\CW}{\textsc{ConfidenceWidth}\xspace}
\newcommand{\playall}{\textsc{PlayAll}\xspace}
\newcommand{\uar}{\textrm{UAR}\xspace}

\newcommand{\cX}{\mathcal{X}\xspace}
\newcommand{\cA}{\mathcal{A}\xspace}

\renewcommand{\Re}{\ensuremath{\mathbb{R}\xspace}}
\newcommand{\R}{\ensuremath{\mathbb{R}}\xspace}
\newcommand{\C}{\ensuremath{C}\xspace}
\newcommand{\A}{\ensuremath{\cA}\xspace}
\renewcommand{\D}{\ensuremath{\mathcal{D}}\xspace}

\newcommand{\F}{\ensuremath{\mathcal{F}\xspace}}
\newcommand{\regret}[1]{\textrm{Regret}(#1)}
\newcommand{\reg}[1]{R(#1)}
\newcommand{\pp}[2]{\pi^{#1}_{#2}}
\newcommand{\ppc}[3]{\pp{#1}{#2|#3}}
\newcommand{\rew}[2]{r^{#1}_{#2}}
\newcommand{\Rew}[2]{\F^{#1}_{#2}}

\newcommand{\up}[2]{\ensuremath{u_{#2}^{#1}}}
\newcommand{\low}[2]{\ensuremath{\ell_{#2}^{#1}}}
\newcommand{\num}[2]{\ensuremath{n_{#2}^{#1}}}

\newcommand{\payp}{\gamma}

 \newcommand{\cs}{\ensuremath{X}\xspace}
\newcommand{\as}{\ensuremath{\hat{X}}\xspace}

\newtheorem{obs}{Observation}
 \newcommand{\eff}{peaked\xspace}
\newcommand{\Eff}{Peaked\xspace}
 \SetKwProg{Fn}{Function}{}{}

\newtheorem{theorem}{Theorem}[section]
\newtheorem{corollary}[theorem]{Corollary}

\newtheorem{lemma}[theorem]{Lemma}
\newtheorem{definition}[theorem]{Definition}

\newtheorem{remark}[theorem]{Remark}

\begin{document}
\title{Fairness Incentives for Myopic Agents}

\author{Sampath Kannan\thanks{University of Pennsylvania. Email: \href{mailto:kannan@cis.upenn.edu}{kannan@cis.upenn.edu}} \and 
Michael Kearns\thanks{University of Pennsylvania. Email: \href{mailto:mkearns@cis.upenn.edu}{mkearns@cis.upenn.edu}} \and 
Jamie Morgenstern \thanks{University of Pennsylvania. Email: \href{mailto:jamiemor@cis.upenn.edu}{jamiemor@cis.upenn.edu}} \and 
Mallesh Pai \thanks{Rice University. Email: \href{mailto:mallesh.pai@rice.edu}{mallesh.pai@rice.edu}} \and 
Aaron Roth \thanks{University of Pennsylvania. Email: \href{mailto:aaroth@cis.upenn.edu}{aaroth@cis.upenn.edu}} \and 
Rakesh Vohra \thanks{University of Pennsylvania. Email: \href{mailto:rvohra@seas.upenn.edu}{rvohra@seas.upenn.edu}} \and
Zhiwei Steven Wu \thanks{University of Pennsylvania. Email: \href{mailto:wuzhiwei@cis.upenn.edu}{wuzhiwei@cis.upenn.edu}}}
\maketitle 

\begin{abstract} 
  We consider settings in which we wish to incentivize
 myopic agents (such as Airbnb landlords, who
  may emphasize short-term profits and property safety) to treat arriving
  clients {\em fairly\/}, in order to prevent overall discrimination against
  individuals or groups. We model such settings in both classical
  and contextual bandit models in which the myopic agents maximize
  rewards according to current empirical averages, but are also
  amenable to exogenous {\em payments\/} that may cause them to
  alter their choices. Our notion of fairness asks that more qualified
  individuals are never (probabilistically) preferred over less qualified
  ones~\citep{joseph2016nips}.

\iffalse
introduces a formal definition of
Consider an agent making choices between $k$
  individuals each day in a contextual bandit setting. Suppose the
  agent is \emph{myopic}: they track empirical estimates of the rewards
  for different groups of individuals, and in round $t$ they choose the individual who
  has had the highest empirically estimated reward. As myopic behavior
  is not in general no-regret, the number of rounds in which the
  myopic agent chooses a sub-optimal individual can be $\Omega(T)$.
  Recent work~\citep{joseph2016nips} introduces a formal definition of
  a bandit algorithm having \emph{fair} behavior: informally, that a
  fair algorithm will never place higher probability on choosing a
  worse individual over a better one. By this definition, then myopic
  play may be unfair in a constant fraction of all rounds. 
\fi

We investigate whether it is possible to design inexpensive
{subsidy} or payment schemes for a principal to motivate myopic
agents to play fairly in all or almost all rounds.  When the principal has
full information about the state of the myopic agents, we show it is possible to
induce fair play on every round with a subsidy scheme of total cost
$o(T)$ (for the classic setting with $k$ arms,
$\tilde{O}(\sqrt{k^3T})$, and for the $d$-dimensional linear
contextual setting $\tilde{O}(d\sqrt{k^3 T})$).  If the principal has
much more limited information (as might often be the case for an
external regulator or watchdog), and only observes the number of rounds
in which members from each of the $k$ groups were selected, but not the
empirical estimates maintained by the myopic agent, the design of such
a scheme becomes more complex. We show both positive and negative
results in the classic and linear bandit settings by upper and lower
bounding the cost of fair subsidy schemes.
\end{abstract}

\vfill
\thispagestyle{empty}
\setcounter{page}{0}
\pagebreak

%!TEX root= paper-outline.tex

\section{Introduction}
Recent uses of machine learning 
to make decisions of consequence for individual
citizens (such as credit, employment and criminal sentencing)
have led to concerns about
the potential for these techniques to be discriminatory or unfair
(\cite{barocas2016big}, \cite{cary16}, \cite{bigdata2}). Existing
research has emphasized discriminatory outcomes originating from
biases encoded into the data sets on which algorithms are
trained.%
\footnote{ For example, Boston's Street Bump program, which
  uses smartphones to determine where road repairs are needed, results
  in certain areas being underserved because of the sparsity of
  smartphones traversing them (\cite{leary}).} 
In this paper, we consider a different source of unfairness in a stochastic bandit
setting. The key friction we examine is when a forward-looking principal
concerned with fairness (such as a regulator or technology platform)
is not the one directly making the choices.  Instead, a
sequence of myopic agents are making the choices. To prevent unfair choices by
these agents, the principal may offer targeted monetary rewards to
agents to incentivize them to make different choices than they would
have in the absence of such payments. Our concern in this paper is how
much the principal needs to be prepared to pay in order to
incentivize fair decisions by myopic agents.

To help fix ideas and motivate our problem, consider a challenge faced
by peer-to-peer (P2P) platforms such as Prosper (P2P lending) or
Airbnb (P2P short-stay housing). The platform cannot dictate to their
users who to extend loans or rent to. Nevertheless, it may wish to
ensure the choices its users make are fair, either to avoid
criticism,\footnote{ Airbnb has very recently received scrutiny over
  both anecdotal reports and systematic studies of racist behavior by
  landlords on their platform (\cite{edelman16}). This study also
  suggests that myopia may play a role in discrimination, in the sense
  that landlords with no prior exposure to minority renters were more
  likely to discriminate.}  or to comply with existing regulations. P2P
lending, for example, is subject to the Equal Credit Opportunity Act
(ECOA). One provision of the Act is a requirement that lenders furnish
reasons for adverse lending decisions. This obligation falls on the
shoulders of the platform, and is challenging to discharge because
the platform aggregates the decisions of many different
lenders.\footnote{Lenders on Prosper must agree to comply with the
  relevant provisions of the ECOA, but there is still evidence of
  discrimination; see \cite{pope}.}

In our model, an agent arrives at each period and must choose amongst
a set of available alternatives. (For instance, the agent might be a lender on Prosper
choosing to whom they will grant a loan, or an Airbnb host choosing
which guest to accept.) We model this as a choice of which arm to
pull in a stochastic bandit setting. We consider both the
classic and contextual bandit cases. In the classic setting, each arm
represents an individual, who will over time repeatedly be considered
for service. In the contextual setting, each arm represents a group,
and individual members of that group are represented by contexts
(i.e. sets of individual features) which change at each round.  A
stochastic reward from the pull of an arm models the uncertain payoff
associated with serving an individual (i.e. extending a loan, or
having the individual rent). Each agent is myopic in the sense that
they are occasional users of this platform, and thus care only about
their current expected payoff. Because of myopia, the agent chooses
the arm with the highest empirical mean in the classical case, and the
context with the highest predicted reward according to a fixed
one-shot learning procedure known to the principal (e.g. ordinary
least squares regression, or ridge regression).  Each agent is limited
to pulling a single arm to model their limited resources (e.g., they
may only have the funds to grant one loan, or host one guest on any
particular night), and of course this simplifies our analysis.

The platform, motivated by a need for ``accountability,'' would like
to be fair. Our formal definition of fairness 
(\citet{joseph2016nips}) 
may be found 
in Section~\ref{sec:prelim},  and can be informally described 
as follows: Suppose an auditor knew the
expected reward of each arm (or more generally, in the contextual
case, of each context), and looked back at the platform's
decisions. Fairness requires that a worse individual was never favored
over a better individual. More precisely, if a platform is fair, then
on any day $t$, the probability $p_x$ that the agent pulls arm $x$ is
such that if the expected reward of $x$ is at least that of $x'$, then
$p_x \geq p_{x'}$. Fairness as defined in this paper does not address
inequities ``outside the model.'' For example, if one group
has lower expected payoff for every context than another, perhaps due to historical inequities, and there are
no additional features available to the learning algorithm, our notion
of fairness permits the agent to favor the higher expected group. In this sense our
fairness notion is aligned with (apparent) meritocracy.

Intuitively, there are two impediments to fairness in this
model. First, neither the platform nor the agents know the
distribution of rewards of each arm. If these were known, the
problem would be trivial: it would be both fair and agent-optimal to always pull the  
arm with highest expected
reward. Second, the agents are myopic --- they have no incentive to
directly invest in fairness or learning.

We examine whether it is possible for the platform, hereafter called
the \emph{principal}, to incentivize the agents to make fair choices
by offering the agent payments for selecting particular arms. These
payments can be randomized.  Because the agents behave identically in
our model (they are all myopic), we treat them as if they are a single
myopic entity, whom we term the \emph{agent}.  We investigate how much
information a principal needs to incentivize fair behavior.
Characterizing the information requirements is important, since
sometimes the principal may be an external regulator or other entity
tasked with oversight without the full information available to the
platform.  At one extreme, called partial information, we suppose the
principal observes only which decisions were made by the agent in each
of the previous rounds but not the rewards. In the P2P context,
rewards might be unobservable because the reward an agent experiences
is a function of both observed characteristics of the borrower or
renter and a private type of the agent. At the other extreme of full
information, the principal has the same information as the agent.

We ask each of these questions in both the classic and linear
contextual bandits settings. In the classic case, individual $i$ is
better than individual $j$ if the mean of distribution $i$ is higher
than that of distribution $j$. In the contextual case, individual $i$
on day $t$ is represented by a set of features, or a context $x^t_i$,
and that individual's expected reward is defined as $f_i(x^t_i)$ for
some $f_i \in C$.

\subsection{Our Results}

\begin{table}
\begin{center}
\begin{tabular}{|c|c|c|}
\hline
& Full Information & Partial Information \\
\hline
Classical & $\tilde{O}(\sqrt{k^3 T})$ cost & $O(\sqrt{T})$ cost for $k=2$\\
 & & $\Omega(T)$ cost for $k \geq 3$ \\
 & & $\tilde{O}(\sqrt{k^3 T})$ cost allowing $\tilde{O}(k^2)$ unfair rounds \\
\hline
Linear Contextual & $\tilde{O}(d\sqrt{k^3T})$ cost & $o(T)$ unfair rounds $\rightarrow \Omega(T)$ cost\\
&  & $o(T)$ cost $\rightarrow \Omega(T)$ unfair rounds\\
\hline
\end{tabular}
\caption{Summary of cost to incentivize perfect fairness and fairness in all but a limited number of rounds.}
\label{tab:summary}
\end{center}
\end{table}

Table~\ref{tab:summary} summarizes our results for characterizing the
cost of incentivizing fair play in myopic agents. Informally, we show
a stark separation: in
the partial information model, any fair payment scheme must cost
$\Omega(T)$; while in the full information model there are payment
schemes which cost $\tilde{O}(d \sqrt{k^3 T})$.  In both the classic
and contextual settings, the full information upper bounds effectively
incentivize the agent to play known fair algorithms. The lower bounds
for the partial information model look somewhat different.  For the
classic setting, our lower bound simply shows that the first round
at which the principal doesn't offer a payment of $1$ will be unfair
with constant probability; for the contextual setting, \emph{every} round
must either be unfair or offer $\Omega(1)$ payment.  

We additionally show that in partial information model, the classic
problem is somewhat easier in two cases. When there are only $2$ arms,
we give a simple payment scheme that only incurs a cost of
$O(\sqrt{T})$ and guarantees with high probability that the agent is
fair at every round. Even when $k \geq 3$, a principal who is allowed
$\tilde{O}(k^2)$ unfair rounds can design a payment scheme which costs
only $\tilde{O}(\sqrt{k^3 T})$.\footnote{In other words, we show that
  by allowing a constant number of unfair rounds, independent of $T$,
  one can achieve sublinear costs.} In the full information model, we
exhibit a payment scheme for the principal which, with high
probability, is fair in \emph{every} round, and has cost
$\tilde{O}(\sqrt{k^3 T})$.

It is interesting to observe that all our payment schemes that
guarantee fairness (either in each period, or except at a constant
number of periods) and achieve sublinear costs also induce the agent
to play so as to experience sublinear regret.  If we think of
sub-linear regret as a proxy for efficiency, our full information
results say we can achieve {\em both} efficiency and fairness with
subsides that grow slowly over time. Under certain conditions, this is
also possible in the partial information setting, which does not
require the principal to ``open the books'' of the agent.

%!TEX root= paper-ec.tex

\subsection{Related Work}
Our work is closely related to the literature on incentivizing
experimentation in bandit settings. In these papers, a sequence of
agents arrives one at a time and each is allowed to select an arm to
pull. Each agent ``lives'' for one period and therefore pulls the arm
that has the highest current estimated payoff given the history. The
agents' myopia means that they do not explore sufficiently and the
patient principal must encourage it.

In this context, \citet{frazier2014incentivizing} explore the
achievable set between the monetary rewards the principal must pay to
incentivize exploration and the time discounted expected reward to the
principal. In this paper the history of actions and outcomes is
observed by the principal and there is, therefore, no examination of
what happens with limited information. More importantly, there is
no consideration of fairness, which is our primary interest.

A different string of papers does not explicitly allow for monetary
payments, but instead the principal discloses information about past
agents' realized rewards to the future agents, see
e.g. \citet{kremer2014implementing}, a more general exploration in
\citet{mansour2015bayesian,mansour16}, or an analytical solution in a
continous time Poisson bandit setting by \citet{che2015optimal} and
the same in a discrete time setting in \citet{papanbimp}. The key
point of departure for our work is that for us, the principal is not
explicitly interested in the long term reward of the agent, but is
instead interested in promoting ``fairness'' (although this will
incidentally have the property of increasing long-term reward of the
agent by encouraging experimentation).

\citet{joseph2016nips} originally proposed our definition of
contextual fairness. They consider the tradeoffs between requiring
this form of fairness and achieving no-regret in both classic and
contextual bandit settings. This notion was also employed in
\citet{joseph2016rawlsian}, when specialized to the linear contextual
bandits case.

%%% Local Variables:
%%% mode: latex
%%% TeX-master: t
%%% End:

%!TEX root= paper-ec.tex

\renewcommand{\H}{\mathcal{H}}
\section{Preliminaries}\label{sec:prelim}
A principal faces a sequence of homogeneous myopic agents, each
operating in a \emph{contextual bandit} setting. For simplicity, we view the agents as a single agent repeatedly making myopic choices.

\subsection{Contextual Bandits}
Let $\{k\}$ refer to a set of arms. In each round $t \in [T]$, an
adversary reveals to the agent a \emph{context} $x_j^t \in \cX$, where
$\cX$ is the common domain of contexts, for each arm $j \in [k]$.  

Fix some class $\C$ of functions of the form $f:\cX\rightarrow [0,1]$.
Associated with each arm $j$ is a function $f_j \in C$, unknown to
both the agent and the principal.\footnote{Often, the contextual
  bandit problem is defined so that there is a single function $f$
  associated with all of the arms. Our model is only more general.}.
An agent who chooses an arm $j$ in period $t$ with context is $x_j^t$
receives a stochastic reward $\rew{t}{j}\sim \Rew{x_j^t}{j}$, where
$\E{\rew{t}{j}} = f_j(x^t_j)$, for some distribution $\Rew{x_j^t}{j}$
over $[0,1]$.

\begin{remark}[Linear Contextual Bandits]
  The results in this paper which involve contexts are for the case
  where the set of contexts $\cX = \{x \in [0,1]^d : ||x|| \leq 1\}$
  for some number $d>0$, and
  $$C = \left\{f_j : \textrm{ there exists some } \theta_j \in [0,1]^d, ||\theta_j|| \leq 1, \textrm{ s.t. }
      f_j(x_j) = \vecdot{\theta_j}{x_j},\,\, \forall x_j \in \cX \right\}.$$
\end{remark}

\begin{remark}[Classic Bandits]
  The \emph{classic} bandits problem is a special case of the
  contextual bandits problem where the set of possible contexts is a singleton. Then $\Rew{\cdot}{j} = \Rew{}{j}$ and
  $\rew{t}{j}\sim \Rew{}{j}$.
\end{remark}

In the running example of P2P lending, $\cX$ represents possible
profiles of attributes that a lender can observe about a potential
borrower. The exact relationship between a profile of attributes at
time $t$, denoted $x_j^t$, and the expected reward earned from
extending a loan to a borrower with this profile is $f_j(x^t_j)$; this
functional form is unknown to the lender, the platform, and the agent.

\subsection{The Myopic Agent and the Principal}
In each period $t$, the principal can offer a vector of payments
$p^t \in \Re_+^k$ to the agent. Here $p^t_i$ is to be interpreted as
the monetary incentive the agent receives from the principal if they
selects arm $i$ on top of any reward that would accrue from the arm
itself.

We assume the agent makes \emph{myopic} choices facing empirical
estimates of the reward from each arm (we describe how these empirical
estimates are constructed momentarily).  We use $\hat{\mu}^t_i$ to
denote the empirical estimated reward for selecting arm $i$ in round
$t$.  When the agent receives a proposed subsidy vector
$p^t \sim \payp^t(\cdot)$, they choose the arm which maximizes the sum
of empirical expected reward and payment, i.e. chooses
$i^t\in \argmax_{i} (\hat{\mu}^t_i + p_i^t)$. This is the sense in
which the agent is myopic---they maximize (their estimate of) today's
net reward plus payment, with no concern for the future. Note that
we assume that rewards are expressed in monetary terms, i.e. that they
are directly comparable to offered payments.

For concreteness and without loss of generality, we fix a tie-breaking
rule---if $M = \argmax_{i} (\hat{\mu}^t_i + p_i^t)$ contains multiple
elements, the agent chooses uniformly at random amongst the members of
$M$ that also maximize payment, i.e.  from the set
$\arg\max_{i \in M} p_i^t$.%
\footnote{Our results do not depend in any important way on the
  particulars of the tie-breaking rule---we chose this one to simplify
  parts of the lower bound proof.}  The principal experiences cost
$p^t_{i^t}$ at round $t$, and total cost $\sum_{t=1}^T p^t_{i^t}$ over
the course of $T$ rounds.  The payments offered by the principal, and
the empirical estimates of the agent, depend on what they know about
past choices and outcomes. We define these next.

\subsection{Information and Histories}
The agent will, in any period $t$, recall  history
$h^t \in \left(\cX^{k} \times \Re_+^k \times [k]\times [0,1] \right)^{t-1} =  \H^t.$ This is a
record of the previous $t-1$ rounds experienced by the agent: $t-1$
4-tuples encoding the realization of the contexts for each arm in a
given period, the payments the principal offered, the arm chosen, and
the realized reward observed.

In the linear contextual case, let $\hat{\theta}^t_i$ represent an
estimate of the linear model $\theta_i$ based on the history $h^t$
(this could, for example, be the ordinary least-squares or regularized
ridge regression estimator --- the important thing is that whichever
method is used is known to the principal). The myopic decision maker,
at day $t$, when facing contexts $x^t_1, \ldots, x^t_k$ will have
empirical estimated reward
$\hat{\mu}^t_i = \vecdot{\hat{\theta}^t_i}{x^t_i}$.  In the classic
setting,
$$\hat{\mu}^t_i = \frac{\sum_{t'< t} \rew{t'}{i}}{|\{t' < t : \textrm{
    $i$ played in round }t'\}|}$$ that is, it represents the empirical
average for the set of previous rewards observed from arm $i$ in
previous rounds.
% (this is also the linear least squares estimator in
%the 0-dimensional case, which the classic bandit setting represents).

Note that during the first several rounds, the myopic reward estimates
$\hat{\mu}^t_i$ are not necessarily defined, e.g. if in the classic
setting, the agent has not yet observed any rewards from arm $i$, or
if in the linear contextual case, the agent has not observed
sufficiently many reward/context pairs to uniquely define the OLS
estimator. To get around this issue, we assume that the agent has
previously observed sufficiently many observations from each arm to
make these estimates well defined --- i.e. at least one observation
per arm in the classic case, and observations corresponding to
contexts that combined form a full rank matrix in the linear case.

We consider two information models for the principal. In the
\emph{full information} model, the principal observes everything the
agent observes, i.e. there is no information asymmetry between the
two. In the \emph{partial information} model, the principal observes
neither the contexts faced by the agents, nor the realized reward of
the arm the agent pulled. We will denote this by
$\underline{h}^t \in \left( \Re_+^k \times [k] \right)^{t-1} =
\underline{\H}^t$.
The principal's information scheme is a function of the information
they have.

In what follows, we define various notions of performance of
algorithm. This is without loss: the payments that the principal
offers, and the resulting choices made by the agents choices, taken
together, can be thought of as an algorithm making choices in a
stochastic bandit setting.

\subsection{Fairness and Regret}
A standard method for measuring the performance of a bandit algorithm
is to measure its \emph{regret}. If one knew $\{f_j\}_{j \in [k]}$,
selecting the arm with highest expected reward in each period would be
optimal. Fix an algorithm $\A$ and let $\pp{t}{}$ be the distribution
over arms at round $t$ of the algorithm: the {\bf regret} of $\A$ is
the difference between the reward of the optimal policy, and the
reward of the agent:
\[ \regret{x^1, \ldots, x^T}= \sum_{t}\max_j\left(f_j(x_j^t)\right) - \Ex{i^t \sim \pp{t}{}}{\sum_t f_{i^t}(x^t_{i^t})} .\]
We say that $\A$ satisfies regret bound $\reg{T}$, if
$\max_{x^1, \ldots, x^T}\regret{x^1, \ldots, x^T} \leq \reg{T}$.

We denote by $\ppc{t}{j}{h^t}$ the probability that $\A$ chooses arm
$j$ after observing contexts $x^t$ in period $t$, given $h^{t}$.  For
economy, we will often drop the superscript $t$ on the history when
referring to the distribution over arms:
$\ppc{t}{j}{h} \coloneqq \ppc{t}{j}{h^t}$.

We now define what it means for an algorithm $\A$ to be fair
in a particular round $t$.  Informally, this will mean that $\A$ will
play arm $i$ with higher probability than arm $j$ in round $t$ only if
$i$ has higher {\em true} expected reward than $j$ in round $t$.

\begin{definition}[Round Fairness]\label{def:fair-round}
  Fix some history $h^t$. Recall $\ppc{t}{j}{h^t}$ is the probability
  that $\A$ plays arm $j$ in round $t$ given the history
  $h^t$.  We will say $\A$ is fair in round $t$ if, for any context
  $x^t$,  for all pairs of arms $j, j' \in [k]$,
  $$\ppc{t}{j}{h} > \ppc{t}{j'}{h}\ \textrm{only if}\
  f_j(x^t_j) > f_{j'}(x^t_{j'}).$$ Similarly, a payment scheme is fair
  in round $t$ if the selection by the myopic agent under the payment
  distribution is fair.
\end{definition}
\begin{remark}
  When this definition is specialized to the classic (noncontextual)
  case, the reward distributions do not vary with time, i.e.
  $\Rew{t}{j} = \Rew{}{j}$ for all $t$. Thus, ``noncontextual'' fairness
  reduces to guaranteeing that if arm $i$ is played with higher
  probability than arm $j$, it must be that the average
  reward drawn from distribution $\Rew{}{i}$ is higher than the
  average reward drawn from distribution $\Rew{}{j}$.
\end{remark}

\begin{remark}
  To be clear about this definition in the partial information model,
  and what we mean by probabilities: note that the ``algorithm'' has
  access to $h^t$ at the beginning of time $t$. By this we mean that,
  the principal has access to $\underline{h}^t$. The principal then
  offers payments to the agent, possibly randomizing. The agent sees
  the full history $h^t$, and the realized payments drawn from a
  distribution, and makes a choice. The principal's randomization,
  and then, if there are ties, the agent's randomization in period
  $t$, can be amalgamated into a net probability of each arm being
  selected in period $t$ after history $h^t$. These are the $\pi$'s
  that the definition refers to.
\end{remark}

We now introduce a notion of fairness which holds at every round with
high probability over the history of observed rewards.

\begin{definition}[Contextual Fairness]\label{def:fair}
  $\A(\cdot)$ is {\bf fair} if, for any input $\delta \in (0,1)$, for
  all sequences of contexts, $x^1, \ldots, x^{t}$ and all reward
  distributions $\Rew{t}{1}, \ldots, \Rew{t}{k}$, with probability at
  least $1-\delta$ over the realization of the history $h$, for all
  rounds $t\in [T]$, $\A(\delta)$ is fair in round $t$.
\end{definition}

Contextual fairness, introduced in~\citet{joseph2016nips}, formalizes
the idea that highly qualified individuals should be treated at least
as well as less qualified individuals.  Here, an individual's
qualification is measured in terms of their expected reward for
$\A$. If two individuals have different profiles (or contexts) but
generate the same expected reward to the learner, this definition
enforces that both be played with \emph{equal} probability every
round. We also introduce a relaxation of contextual fairness, which
allows for an algorithm to have some number of unfair rounds.

\begin{definition}[$g$-Contextual Fairness]\label{def:fair_contextual-g}
  $\A(\delta)$ is $g$-{\bf fair} if, for any input $\delta \in (0,1)$,
  for all sequences of contexts, and all reward distributions, with
  probability at least $1-\delta$ over the realization of the history
  $h$, for all but $g$ rounds $t\in [T]$, $\A(\delta)$ is fair in
  round $t$.
\end{definition}

Our principal is willing to incentivize the agent's behavior to ensure
contextual fairness (and, incidentally, low regret). We investigate
what subsidy schemes incentivize fair choices by a myopic agent, and
the cumulative cost of such subsidies. We show this answer depends
upon the kind of information the principal has access to:
incentivizing fair play with partial information is in general very
expensive, while incentivizing fair play under full information need
not be so.

%%% Local Variables:
%%% mode: latex
%%% TeX-master: t
%%% End:

%!TEX root= paper-ec.tex

\section{A Principal with Partial Information Cannot Ensure $T$ Fair Rounds}\label{sec:classic-lower}

In this section, we give a lower bound on the total payments needed in
the partial information setting to ensure contextual fairness in every
round. In fact, we don't even need to move to the contextual case:
this section focuses on the classic bandit setting (where the context
$x^t_j$ is invariant with respect to $t$ for each arm $j$). We show
that in the partial information setting, any principal who
incentivizes a myopic agent to satisfy contextual fairness in each
round must incentivize uniformly random play in each of the $T$
rounds, which has cumulative cost $\Omega(T)$.

\begin{theorem}\label{thm:lower-partial}
  Suppose $k \geq 3$. There is an instance such that any fair payment
  scheme in the partial information model must, with probability
  $1-\delta$, (where $\delta$ is the fairness parameter passed to the
  principal) spend $\Omega(T)$ in payments over $T$ rounds and incur
  regret $\Omega(T)$.
\end{theorem}

The lower bound proceeds from the following idea: at the first round,
the principal has no information about what the instance is. Hence, in
order to guarantee fairness against all instances, they must proceed
cautiously and use a payment scheme that is able to induce uniformly
random play (the only distribution that is fair for all instances) for
every possible realization of empirical means. Because empirical means
can range between 0 and 1, this will cost $1$. However, because this
payment distribution (by design) induces identical behavior on every
possible instance, it does not allow the principal to learn anything
about the instance. Thus, in every round before which fair play has
been guaranteed, the principal has the same informational
disadvantage. By induction, therefore, they must induce uniformly random
play at every round, at a cost of 1 per round.

We show that fairness at every round, against all instances is equivalent to the
payment scheme in each round being what we term \emph{\eff.} A \eff
payment rule is one that can always incentivize the play of some arm
regardless of the empirical means the myopic agent currently has. This is
equivalent to saying that there is some arm $i\in[k]$ for which
$p_i \geq  p_{i'} + 1$ for all $i' \neq i$. This will imply the payment
scheme must spend $\Omega(1)$ in each round to incentivize fair play,
or $\Omega(T)$ in total.

\begin{definition}[\Eff] \label{def:effective} Let $p \in \R^k$. If
  for some $i \in [k]$, $p_i \geq \max_{i' \neq i} p_{i'} + 1$, we say
  $p$ is \emph{\eff}. If
  \[\Pr{p\sim \D}{p\textrm{ is \eff}} = 1\]
  then we call distribution $\D$ \eff.
\end{definition}

\begin{obs}\label{obs:zero}
  If a principal uses a \eff distribution $\D$ in a round, they learn
  nothing about the instance the agent faces from the agent's play in
  that round.
\end{obs}
\begin{proof}
  By Definition \ref{def:effective}, every payment scheme drawn from
  $\D$ is \eff. In other words, for every $p_\ell \sim \D$: there is
  some $i_{\ell}$ such that
  $p_{i_{\ell}} \geq \max_{i' \neq i_{\ell}}p_{i'} +1$.
  Thus, the myopic agent will choose $i_\ell$ when presented with
  $p_\ell$ regardless of the instance the agent faces.
\end{proof}

The main idea behind Theorem~\ref{thm:lower-partial} is in proving
that any fair payment scheme must be \eff in every round.
Technically, we use the fact that the principal learns nothing about
the instance from a \eff distribution to allow us freedom to design a
lower bound instance as a function of the first distribution $\D^t$
deployed by the principal that is not \eff.  Because this
distribution, by virtue of being the first non-\eff distribution,
cannot be a function of the underlying instance, we are unconstrained
in our ability to tailor the instance as a function of $\D^t$. We then
show this instance forces an unfair round for $\D^t$; we can conclude
that with probability $1-\delta$, the principal must \emph{never}
deploy any distribution over payments that is not \eff.

\begin{lemma}\label{lem:effective}
  For any fairness parameter $\delta$, a fair payment scheme must with
  probability $1-\delta$ generate a sequence of payment distributions
  $\D^1, \ldots, \D^T$ such that each $D^t$ is \eff.
\end{lemma}

We now conclude the proof of Theorem~\ref{thm:lower-partial} before
presenting the proof of Lemma~\ref{lem:effective}.
\begin{proof}
  Lemma~\ref{lem:effective} implies that a fair payment scheme must
  with probability $1-\delta$ generate $T$ \eff payment
  distributions. Since $\max_i p_i \geq \max_{j \neq i} p_j + 1$, and
  $\hat{\mu}^t_i \in [0,1]$ for all $i$, the payment scheme's largest
  payment is always at least $1$, and is always accepted. Thus, the
  myopic agent will receive a payment of at least $1$ in every round, for
  a total cost of $\Omega(T)$.

  To prove the regret of this payment scheme may be $\Omega(T)$ on
  some instances, consider each of the $k$ instances in which one arm
  has mean $1$ and the remaining $k-1$ arms have mean $0$. By
  Observation \ref{obs:zero}, the principal has no information about
  which of these instances is realized. Therefore to be to be fair
  with respect to all of these instances, each arm must be assigned
  the largest payment with equal probability, which induces uniformly
  random play amongst all $k$ arms, $\Omega(1)$ regret per round, and
  cumulative regret $\Omega(T)$.
\end{proof}

We now present the proof of the main lemma for this section: that in
order for a payment distribution to be fair, it must be
\eff. Informally, we first show that any fair payment distribution
must be ``invariant under permutation'': any coordinate $i$ should
have have an equal probability of having the largest payment, and have
an equal probability of $j$ being the second-largest payment with
margin $c$, for each value of $j$ and $c$. We then show in the first
round $t$ at which the payment distribution is un\eff, $\D^t$ is
unfair for some instance $I$ constructed as a function of $\D^t$.

\begin{proof}[Proof of Lemma~\ref{lem:effective}]
  We consider some round $t$. Suppose that for every round $t' < t$,
  the payment distribution $\D^{t'}$ was peaked. If the payment
  distribution $\D^t = \D$ at round $t$ is fair, we show that it too
  must be \eff. Observe that by Observation \ref{obs:zero}, $\D$ must
  be defined independently of the underlying instance $I$, and because
  fairness is defined in the worst case over instances, we continue to
  have complete freedom in choosing $I$.

  We first claim that in round $t$, if $\D$ is fair, for any two
  distinct $i,i', \in [k]$ and any $c \in [0,1]$, that
\begin{equation}\Pr{p \sim \D}{p_i \geq \max_{\ell \in [k], \ell \neq i}p_\ell + c} = \Pr{p \sim \D} {p_i' \geq \max_{\ell\in [k], \ell \neq i'}p_\ell + c}. \label{eqn:same-gap}\end{equation}
Suppose Equation~\ref{eqn:same-gap} does not hold. We construct an
instance for which $\D$ will not be fair, a contradiction. Suppose the
left-hand side is larger than the right-hand side. Consider an
instance where $\mu_i = \mu_{i'} = 1-c$, and all other arms (of which
there is at least $1$) have means $\mu_j = 1$. Suppose further that
the distribution over $i$'s reward is deterministic point mass at
$1-c$, whereas $i'$'s reward distribution yields reward
$1-c + \epsilon$ with probability $\frac{1}{2}$ and $1-c-\epsilon$
with probability $\frac{1}{2}$. Then, with probability at least
$\frac{1}{4}$, $\hat{\mu}_{i'} < \hat{\mu}_i$.%
\footnote{For $t$ odd it is $\frac{1}{2}$, for even $t$ it is
  $\frac{1}{2}\left(1- \cdot{ {t}\choose{t/2}}\right) \cdot \frac{1}{2^t} \geq
  \frac{1}{4}$,
  achieved at $t=2$ and increasing in $t$.}  Thus, with probability
$\frac{1}{4}$ over the history of rewards observed, $i$ wins with
higher probability than $i'$, since $i$ wins whenever $i$'s payment is
the largest by $c$, and $i'$ can only win when $i'$'s payment is the
largest by at least $c$ for any history for which
$\hat{\mu}_{i'} < \mu_{i'}$. This is a violation to fairness for
$\delta < \frac{1}{4}$.

Notice that Equation~\ref{eqn:same-gap} implies that each arm receives
the highest payment with probability $\frac{1}{k}$, and that this also
holds conditioning on any gap $c$ between highest and second-highest
payments.

Now, since $\D$ is not \eff,
\[\Pr{p \sim \D}{\exists i : p_i \geq \max_{i' \neq i} p_{i'} + 1} < 1.\]
Define $c$ as follows:
\[c = \sup_{y \geq 0} \text{s.t. }\Pr{p \sim \D}{\exists i : p_i \geq \max_{i' \neq i} p_{i'} + y} = 1.\]
Notice, this implies that:
\begin{align}
\forall \epsilon>0, \exists \eta >0 \text{ s.t. }\Pr{p \sim \D}{\exists i : p_i \geq \max_{i' \neq i} p_{i'} + c +\epsilon} \leq 1- \eta. \label{eqn:prob}
\end{align}
We now construct an instance as a function of $c$. There are two cases
-- either $c>0$ or $c=0$.

\paragraph{Case 1: $c>0$}
Consider the following instance, defined in terms of $c$ and a
constant $0 < \epsilon < c$: arm $1$ has mean $1-c$ with deterministic
rewards, arm $2$ has mean $1-c$ with reward $1-c-\epsilon$ with
probability $\frac{1}{2}$ and reward $1-c+\epsilon$ with probability
$\frac{1}{2}$, and arms $3,\ldots, k$ have a deterministic reward of
$1$. Note that by definition of $c$, and the deterministic nature of
arm 1's distribution, we have that for every history $h$,
$\pi^t_{1|h} \geq 1/k$. By the fairness constraint, we must therefore
also have that for every other arm $i > 1$, $\pi^t_{i|h} \geq 1/k$,
since no other arm has lower mean. This implies that for every arm
$i$, it must be that $\pi^t_{i|h} = 1/k.$

Note, as we argued in footnote $7$, that with probability at least
$\frac{1}{4}$, for any $t$, $\hat{\mu}_{2} < \hat{\mu}_{1} = 1-c$, by
construction. In this case, there is some $\epsilon' > 0$ such that
arm $2$ is not played unless
$p_2 > \max_{i \neq 2} p_i + c + \epsilon'$. However, by definition of
$c$, this occurs with probability strictly less than $1/k$,
contradicting the assertion that $\D$ is a fair distribution.

\paragraph{Case 2: $c=0$}
Consider the instance in which arms $1, \ldots, k-1$ have mean
$\frac{1}{2}$ and deterministic reward distributions, while arm $k$
has mean $1/2$, and stochastic rewards that are
$\frac{1}{2} - \epsilon$ with probability $\frac{1}{2}$ and
$\frac{1}{2}+ \epsilon$ with probability $\frac{1}{2}$. Note that in
this case, fairness requires that all arms be played with identical
probabilities. With probability at least $\frac{1}{4}$, arm $k$ has
empirical mean lower than its true mean.  Condition on
$\hat{\mu}_k < \mu_k$. In this case, since $c = 0$, with arm $k$ must
be selected with probability less than $\frac1k$ since the payment to
arm $k$ will be strictly less than $\mu_k - \hat{\mu}_k$ with strictly
positive probability (\ref{eqn:prob}), and therefore unfair.
\end{proof}

\subsection{A Fair Payment Scheme for Two Arms}

We now show that having at least three arms
is necessary for our lower bound result. Indeed, in the classic stochastic partial information setting with two arms there exists a
simple payment scheme that can ensure fairness in every round while
achieving sublinear regret and payment.

The key idea in this payment scheme is to maintain confidence
intervals around empirical reward means for the two arms.  The
following lemma tells us how to construct
confidence intervals.
\iffalse \swcomment{maybe makes more sense to move the following
  to prelim} \mpcomment{I'm happy with it here at this point.}\fi
  \begin{lemma}[Lemma 1, \citep{joseph2016nips}]\label{lem:width}
  Suppose arm $i$ has been pulled $\num{t}{i}$ times before round
  $t$. Let
  $\low{t}{i} = \hat{\mu}^t_i -\sqrt{\frac{\ln\frac{(\pi
        (t+1))^2}{3\delta}}{2\num{t}{i}}}$,
  and
  $\up{t}{i} = \hat{\mu}^t_i + \sqrt{\frac{\ln\frac{(\pi
        (t+1))^2}{3\delta}}{2\num{t}{i}}}$.
  Then, with probability at least $1-\delta$, for every
  $i\in [k], t\in[T]$, $\low{t}{i} \leq \mu_i \leq \up{t}{i}$.
\end{lemma}

In the light of this result, we will define the function \CW as
follows, which will also be useful for describing our payment scheme
in the following section.
\begin{equation}
  \CW(\delta, t, n) = 2\sqrt{\frac{\ln\left(\frac{(\pi \cdot (t+1))^2
        }{3\delta}\right)}{n}}
\end{equation}
Given this confidence width function, our payment scheme is the
following: in each round $t$, choose an arm $a^t$ uniformly at random,
and offer payment $p(\delta, t, n_1^t, n_2^t)$ for playing arm $a^t$
and offer 0 for playing the other arm, where
\[
  p(\delta, t, n_1^t, n_2^t) = \CW(\delta, t, n_1^t) + \CW(\delta, t,
  n_2^t)
\]
and $n_1^t, n_2^t$ denote the number of times that the two arms are
played before round $t$. Whenever the agent selects the arm associated
with zero payment, the principal will then offer zero payment for both
arms in all future rounds.

% \begin{center}
% 	\begin{algorithmic}[1]
%           \Procedure{\CW}{$\delta, t,  n$}
%           \Return $2\sqrt{\frac{\ln\left(\frac{(\pi \cdot t)^2 }{3\delta}\right)}{n}}$
%           \EndProcedure
%  	\end{algorithmic}
% \end{center}

\begin{theorem}\label{thm:lower-partial-twoarms}
  Consider the classic case with $k = 2$ arms in the
  partial information setting. Then the payment scheme above
  instantiated with parameter $\delta$ is fair in every round with
  probability at least $1 - \delta$. Moreover, the incurred total cost
  and expected regret are at most $\tilde O(\sqrt{T})$.
\end{theorem}

%\begin{proof}[Sketch of Proof]
%  (Placeholder for now) Force the agent to play UAR for the first
%  ($O(T^{2/3})$) rounds using a \eff payment rule. The agent
%  therefore has played each arm $O(T^{2/3})$ times. The confidence
%  intervals are therefore of size $O(T^{-1/3})$. Therefore can ensure
%  fair play in the remaining rounds for a cost of $T^{2/3}$ as
%  well. Since the confidence intervals are small, expected regret is
%  small.
%\end{proof}

%%% Local Variables:
%%% mode: latex
%%% TeX-master: t
%%% End:

\section{Classic Setting: 
Sublinear Payments with Only $\tilde{O}(k^2)$ Unfair Rounds}

The necessity of linear growth in subsidies (Theorem
\ref{thm:lower-partial}) was driven by the requirement that the agent
satisfy contextual fairness in each period. It is natural to ask what
would happen if one relaxed this requirement. In this section, we
describe how to design a payment scheme which will satisfy contextual
fairness in all but $\tilde{O}(k^2)$ rounds. We show that it is
possible to achieve payments and regret which grow sub-linearly with
$T$.

The rough idea behind this upper bound is inspired
by~\citet{joseph2016nips} who show that fairness can be achieved by
maintaining confidence intervals around empirical arm means, and
enforcing the constraint that any pair of arms with overlapping
confidence intervals are played with equal probability: in particular,
a fair no-regret algorithm can play uniformly at random amongst the
set of arms ``chained'' to the arm with highest upper confidence bound
by the confidence intervals, called the \emph{chained set} \cs.

\iffalse % move to 2arm
The following lemma from ~\citet{joseph2016nips} tells us how to
construct confidence intervals.
  \begin{lemma}[Lemma 1, \citep{joseph2016nips}]\label{lem:width}
  Suppose arm $i$ has been pulled $\num{t}{i}$ times before round
  $t$. Let
  $\low{t}{i} = \hat{\mu}^t_i -\sqrt{\frac{\ln\frac{(\pi
        (t+1))^2}{3\delta}}{2\num{t}{i}}}$,
  and
  $\up{t}{i} = \hat{\mu}^t_i + \sqrt{\frac{\ln\frac{(\pi
        (t+1))^2}{3\delta}}{2\num{t}{i}}}$.
  Then, with probability at least $1-\delta$, for every
  $i\in [k], t\in[T]$, $\low{t}{i} \leq \mu_i \leq \up{t}{i}$.
\end{lemma}\fi

Denote the confidence interval associated with arm $i$ at round $t$ by
$[\low{t}{i}, \up{t}{i}]$. Fix a set of confidence intervals at round
$t$, $[\low{t}{1}, \up{t}{1}], \ldots, [\low{t}{k}, \up{t}{k}]$. We
say $i$ is \emph{linked} to $j$ if
$[\low{t}{i}, \up{t}{i} ]\cap [\low{t}{j}, \up{t}{j}] \neq \emptyset$,
and $i$ is \emph{chained} to $j$ if $i$ and $j$ are in the same
component of the transitive closure of the linked relation.  We refer
to the set of arms chained to the arm with highest upper confidence
bound as the \emph{chained set} \cs. We say the sequence of confidence
intervals are \emph{valid} if, with probability $1-\delta$, they
contain the true and empirical averages of every arm in every round.

In the absence of explicit knowledge of the sample means, the
principal does not have sufficient information to incentivize
uniformly random play amongst exactly the set of arms chained to the
arm with highest upper confidence bound \cs\footnote{Indeed, the ability to do so would contradict
  Theorem~\ref{thm:lower-partial} by acheiving sublinear regret with
  zero unfair rounds.}. The principal does not know the empirical
means of the arms, and therefore cannot compute the arms contained in
\cs directly.

The principal can, however, incentivize the myopic agent to play an
arm $j$ with a payment vector $p^t$ such that
$p^t_j \geq \max_i p^t_i + |\max_i\hat{\mu}^t_i - \hat{\mu}^t_j|$.
Unfortunately, the principal neither knows \emph{which} arms belong to
$\cs$, nor how many arms are in $\cs$, nor how far apart the empirical
means are in $\cs$. Instead, the principal can maintain upper bounds
on all of these quantities. Namely, the principal tracks a superset of
the chained set \cs, called the \emph{active set} $\as$. $|\as|$ will
then act as an upper bound on the size of the chained set, and
$|\as| \cdot x_{\as}$ will upper bound the difference between
the highest arm mean and the lowest chained arm's means, where $x_{\as}$ is the
width of the largest confidence interval of any arm in $\as$. By
offering a payment of $|\as| \cdot x_{\as}$ to an arm selected
uniformly from $\as$ (and zero for all other arms), the principal will
cause uniformly random play amongst $\as$ if all arms in $\as$ have
empirical means within $|\as| \cdot x_{\as}$ of the best empirical
mean.

This is fair if in every round $\as = X$: all means will then be
within $|\as| \cdot x_{\as}$ by the definition of chaining and
$x_{\as}$, and so this will induce uniformly random play amongst the
chained set, exactly the behavior shown to be fair
in~\citet{joseph2016nips}.  On the other hand, if
$\as \setminus X \neq \emptyset$, this behavior could be unfair,
either because not all arms within \as have empirical means within
$|\as| \cdot x_{\as}$ of one another (i.e., not all arms in the set
are chained together), or because some arms in $\as$ chain to other
arms outside of \as, or because some arms in \as are ``below'' arms
outside of \as. We will guarantee the latter issues do not occur, by
always ensuring \as contains any arms ``above'' or chained to any arm
in \as.  The former issue (that some arms in \as may not be chained to
others in \as, and their empirical means may then not be close enough
for the payment to change the myopic agent's behavior in all cases)
cannot be entirely avoided.  However, we can quickly discover if any
arm in $\as$ has empirical mean less than $|\as| \cdot x_{\as}$ below
the best empirical mean: in $O(\as) = \tilde{O}(k)$ rounds, that arm will be
offered the subsidy and it won't change the agent's decision. Those
$\tilde{O}(k)$ rounds will be unfair, as are several rounds which follow this
discovery and update the set $\as$.

The following lemma, a generalization of the analysis
of~\citet{joseph2016nips}, can be interpreted to mean the
following. Fix a definiton of confidence intervals which are all valid
over all rounds for all arms with probability $1-\delta$.  Consider
any set of arms $S$ which (a) contains the ``upper chain'' (all arms
chained to the arm with highest upper confidence bound), (b) contains
any arms ``above'' the confidence intervals of any arm in the set, and
(c) is closed under chaining. Then, playing uniformly at random
amongst $S$ will satisfy contextual fairness.

  \begin{lemma}\label{lem:intervals}
    Suppose, with probability $1-\delta$, at every round $t$ and for
    every arm $i$, $\mu^t_i \in [\low{t}{i}, \up{t}{i}]$. Consider
     a set $S$ of  arms with the $k'$ highest upper confidence bounds for some $k'<k$. Then, it is fair to play uniformly at
    random over $S \cup \{i\textrm{ chained to an arm in }S\}$.
\end{lemma}

The pseudo-code in Figure~\ref{fig:playall} describes the payment
scheme, which we analyze thereafter.
\begin{algorithm2e}[h!]
\begin{varwidth}[t]{0.55\textwidth}        % change 0.45 to suit your need
\Fn{\playall($\delta, T$)}{
  $x \gets 1$\;
  $\as \gets \{1,\ldots,k\}$\;
  \While{$t \leq T$} {
    $(x,\as) \gets \chainfair(\delta, x, \as)$\;
  }
}

\Fn{ \chainfair($\delta, x, \as$)}{
  Choose $j^t \in_{\uar} \as$ \tcp*{Pick arm to incentivize}
  $x \gets \min(x, \CW(\delta, t, \min_{j\in \as} n^t_{j}))$\;
  Offer $p^t$ : $p^t_{j^t} = 4 x \cdot |X|, p^t_{i' \neq j^t} = 0$\;
  $i^t \gets$ the myopic player's choice\;
  \If{$i^t \neq j^t$} {
    $\as \gets \findchain(x, \as,  t)$\;
  }
  \Return $(x, \as)$
}
\end{varwidth}\quad\quad
\begin{varwidth}[t]{0.4\textwidth}
\Fn{\findchain($x, \as,  t$)}{
 	 Offer $p^t = \vec{0}$\;
         $i^t \gets$ the myopic player's choice\;
         $R \gets \{i^t\}$\;
         Offer  $p^t = p^{t-1} + 2 \cdot x \cdot \sum_{i\in \as \setminus R}   e_{i}$\;
        \While{$i^t \gets$ the myopic player's choice and $i^t \notin R$} {
         $R = R \cup \{i^t\}$\;
         $t \gets t + 1$\;
         Offer  $p^t = p^{t-1} + 2 \cdot x \cdot \sum_{i\in \as \setminus R}   e_{i}$\        \tcp*{add $2 \cdot x$ to  payments of arms in $\as$ not yet chosen}
       }
        \Return $R$
}
\end{varwidth}\quad\quad
\caption{ $O(k^2 \ln \frac{k}{\delta})$-fair Payment Scheme \label{fig:playall}}
\end{algorithm2e}

%\iffalse % MOVED to 2arm
%          \Procedure{\CW}{$\delta, t,  n$}
%          \Return $2\sqrt{\frac{\ln\left(\frac{(\pi \cdot t)^2 %}{3\delta}\right)}{n}}$
%          \EndProcedure
%
%\end{algorithm2e}
%\fi

The performance of this payment scheme is summarized in the following
theorem.
\begin{theorem}\label{thm:partial-upper}
For any $\delta$,
  \playall is $O(k^2\ln(k/\delta))$-fair, and has expected cost and regret
  \[O\left(k \cdot \sum_{t}\CW(\delta, t, \frac{t}{k})\right) =
  O\left(\sqrt{k^3T \ln\frac{T}{\delta}}\right).\]
\end{theorem}
We present the proof to this theorem after stating several lemmas
describing the behavior of \playall.  Observation~\ref{obs:confidence}
states that using $x$ as a confidence interval width for all arms in
$\as$ yields valid confidence intervals. Hereafter, we use
$[\low{t}{i}, \up{t}{i}] = [\hat{\mu}^t_i - x, \hat{\mu}^t_i + x]$ as
valid confidence intervals for all $i\in \as$, $t\in [T]$.
Lemma~\ref{lem:findchain-upper} shows that \findchain outputs a set
which contains the upper confidence chain in its output round.
Lemma~\ref{lem:findchain-chained} states that \findchain's output is
closed under chaining (e.g., that every arm in its output is only
chained to arms also belonging to the output set) and contains all
arms ``above'' any arms in its
output. Lemma~\ref{lem:findchain-problem} argues that the empirical
means of every arm in the set output by \findchain are within $4$
confidence interval widths of some other arm in the set.
Lemma~\ref{lem:chainfair-uniform} shows that when this is the case
(that the empirical means are within $4x$ of each other, as is the
case right after a call to \findchain), that \chainfair induces
uniformly random play amongst $\as$. Lemma~\ref{lem:chainfair-find}
upper-bounds the number of rounds before which \chainfair will
discover when it is inducing unfair play.  All proofs of these lemmas
can be found in Section~\ref{sec:missing}.
\begin{obs}\label{obs:confidence}
  With probability $1-\delta$, for all $t\in [T]$, $i \in \as^t$,
  $\mu_i \in [\hat{\mu}^t_i - x, \hat{\mu}^t_i + x]$.
\end{obs}
\begin{lemma}\label{lem:findchain-upper}
  $\findchain(x, \as, t)$ contains all arms chained to the arm with
  highest upper confidence bound in its output round $t'$.
\end{lemma}

\begin{lemma}\label{lem:findchain-chained}
  Any arm chained to the set $R = \findchain(x, \as, t)$ belongs to
  $R$. Moreover, any arm $i \notin R$ must have
  $\up{t}{i} < \min_{i' \in R} \low{t}{i'}$.
\end{lemma}

\begin{lemma}\label{lem:findchain-problem}
  Let $t'$ be the round in which $R = \findchain(x,\as,t')$ outputs
  $R$. Then, for any $j \in R = \findchain(x,\as, t')$,
\[\hat{\mu}^{t'}_j \geq \min_{j' \in R\setminus \{j\}} \hat{\mu}^{t'}_{j'} - 4 \cdot x.\]
Moreover, $\max_{j\in R} \hat{\mu}^{t'}_j - \min_{j\in R} \hat{\mu}^{'t}_j \leq (2 | R|  + 2)\cdot x$.
\end{lemma}

\begin{lemma}\label{lem:chainfair-uniform}
  Suppose
  $\max_{i,j \in \as}|\hat{\mu}^t_{i} - \hat{\mu}^t_j| \leq 4 |\as|
  \cdot x$.
  Then $\chainfair(\delta, \as)$ induces uniformly random play amongst
  $\as$.
\end{lemma}

\begin{lemma}\label{lem:chainfair-find}
  Whenever there is an arm $i$ such that
  $\max_{j \in \as}|\hat{\mu}^t_{j} - \hat{\mu}^t_i| > 4 |\as| \cdot
  x$,
  with probability $1-\delta$,
  \findchain will be called within $O(k\cdot\ln(1/\delta))$ many rounds.
\end{lemma}

\begin{proof}[Proof of Theorem~\ref{thm:partial-upper}]
  We first upper-bound the number of rounds in which \playall might
  violate the fairness condition.

  We argue iteratively about the set $\as$: that (a) all arms chained
  to the top arm belong to $\as$, and (b) all arms chained to any arm
  in $\as$ belong to $\as$. This is trivially true initially as
  $\as = \{1,\ldots,k\}$. $\as$ is only updated as the result of a call to
  \findchain.  By Lemma~\ref{lem:findchain-upper}, any arm chained
  to the top arm will remain in \as. Furthermore, by
  Lemma~\ref{lem:findchain-chained}, any arm chained to an arm in its
  output also belongs to its output. Thus, (a) and (b) hold for $\as$
  for all rounds.

  So, in rounds in which \chainfair induces uniformly random play
  amongst \as, (a) and (b) imply \chainfair satisfies the fairness
  condition. For any round in which
  $\max_{i,j\in\as}|\hat{\mu}^t_i - \hat{\mu}^t_j| \leq 2 |\as| \cdot
  x$,
  Lemma~\ref{lem:chainfair-uniform} implies \chainfair induces
  uniformly random play amongst \as. By
  Lemma~\ref{lem:findchain-chained}, \as contains any arms either
above or chained to arms in \as. Thus, Lemma~\ref{lem:intervals}
applies, and these rounds are fair.

We now upper-bound the number of rounds for which \chainfair does not
induce uniformly random play amongst \as. For any particular $i$ and
round $t$ such that
$\max_{j\in\as}|\hat{\mu}^t_j - \hat{\mu}^t_i| > 4 |\as| \cdot x$,
Lemma~\ref{lem:chainfair-find} implies that this will be found in
$O(k\ln(1/\delta))$ rounds, and \findchain will be called. In any future round
$t' \geq t$, since the confidence intervals are valid, we know that
$\max_{j\in\as}|\hat{\mu}^{t'}_j - \hat{\mu}^{t'}_i| > 4 (|\as|
-2)\cdot x$,
since either of the two means can change but by at most $x$ each.
Lemma~\ref{lem:findchain-problem} will return $\as$ such that
$\max_{i,j\in\as}|\hat{\mu}^t_i - \hat{\mu}^t_j| \leq (2 |\as| + 2)
\cdot x$.
Thus, as $|\as| \geq 2$, then arm $i$ will be removed at the first
round in which it was the impetus for \findchain to be called as
$\max_{i,j\in\as}|\hat{\mu}^{t'}_j - \hat{\mu}^{t'}_i| \leq 2 (|\as|
+2)\cdot x \leq 4 (|\as| -2)\cdot x < \max_{j\in\as}|\hat{\mu}^t_j -
\hat{\mu}^t_i| $, a contradiction if $i\in \as$.

  Since $x$ is non-increasing, so is $\as$: thus, at most
  $k$ calls to \findchain are made. Thus, the total number of unfair
  rounds is equal to the number of rounds in which
  $\max_{i,j\in\as}|\hat{\mu}^t_i - \hat{\mu}^t_j| > 4 |\as| \cdot x$
  plus the number of rounds in \findchain. The former is bounded by
  $O(k^2\ln(k/\delta))$ (With probability $1-\delta/k$ it will take at most $O(k\ln(k/\delta))$ rounds of unfair play before
  \findchain is called when this is the case, and each call will
  reduce the size of \as so it can be called at most $k$ times. In total, this bound holds for all $k$ rounds with probability $1-\delta$.); the
  latter by $O(k^2)$ (each call of \findchain uses $O(k)$ rounds, and
  there are at most $O(k)$ calls to \findchain).

  We now upper-bound the cost of this payment scheme and the regret of
  the agent. In the $O(k^2\ln(k/\delta))$ unfair rounds, the payments might be
  $\Omega(1)$; similarly, the regret of the algorithm in those rounds
  might be $\Omega(1)$. In all other rounds, the myopic agent is playing
  uniformly at random amongst a set of arms whose true means are
  within $2 k \cdot x$ of the best true mean, so $2k \cdot x$ in each
  fair round is an upper-bound on per-round regret.  The maximum
  payment offered in any round is $4 k \cdot x$ as well, so that also
  upper bounds the cost. The overall upper bound follows from some
  basic algebra and the fact that each arm in $\as$ will have been
  played $\tilde{\Omega}\left(\frac{t}{k}\right)$ times in round $t$.
\end{proof}

%\subsection{Fairness in all rounds when principal has full information}\label{sec:classic-full}
%Here we consider a principal who has full information about both the
%number of times each arm was pulled in round $t$ and the rewards
%earned by the myopic agent in each prior round. That is, a
%\emph{partial information payment scheme} is defined to be a
%distribution $\G$ over sequences of policies
%%
%\[\Payp= (\payp^1, \ldots, \payp^T), \payp^t : ([k] \times [0,1])^t
%\to \dist
%\]
%%
%which map from the history of which arm was pulled in each round and
%the reward for that pull to distributions over payment
%vectors. Redefine $\ppc{t}{h}{j}$ to be the (expected) probability of
%a myopic agent receiving payments $p \sim \payp^t(h^t)$ breaking ties as
%above, where the expectation is over the randomness in picking
%$\Payp$. Then, a \emph{fair full information payment scheme} is a
%distribution $\G$ such that a myopic agent reacting to payments
%$p\sim \payp^t(h^t)$ will produce $\ppc{t}{h}{i}$ and
%$\ppc{t}{h}{j}$ (as defined above) which satisfy
%Definition~\ref{def:fairness} for every instance $I$.

\newcommand{\argmin}{\textrm{argmin}}
\section{Contextual Setting with Partial Information:\\
  Linear Payments or Unfair Rounds}
In this section, we argue that the partial information model is much
harder in the linear contextual case --- in \emph{every} round that
the principal does not pay $\Omega(1)$, an adversary can force the
myopic agent to behave unfairly. This implies that on an adversarially
chosen instance, every round is either unfair or has constant cost:
thus, either the sum of the payments must be $\Omega(T)$, or the
number of unfair rounds must be $\Omega(T)$, or both. This rules out
positive results in the partial information model of the sort we were
able to obtain in the classic bandits setting. In the following, we
assume that the myopic agent is using an ordinary least squares
estimator, for simplicity. Identical results can also be proven for
other natural estimators, like ridge regression estimators.
\begin{theorem}\label{thm:partial-lower-contextual} Suppose $k \geq 3$. 
  Consider any payment scheme in the partial information model in the
  linear contextual bandit setting.  For any $\eta \in (0,1)$, there
  is an instance for which with probability $1-\delta$, in every
  round, either the round is unfair, or the expected cost for the
  principal is $\frac{k-1}{k}\cdot(1-\eta)$.
\end{theorem}

The proof of the theorem relies on the fact that the principal cannot
observe the adversarially chosen contexts; the expected rewards in any
round then can be (almost) arbitrary. In the classic case, it was only
in the first unpeaked round that we had the freedom to design our
lower bound instance arbitrarily -- after that, the principal would
have learned some information about the instance, and hence the
payment distribution could be a function of the instance. In the
linear contextual case, we have sufficient freedom to design a lower
bound instance at -every- round. Although the principal may have
learned a great deal about the underlying linear functions, she by
definition has no information about the realized contexts at the
current round, which we use to our advantage.  As in the classic
setting, in any round where the payment scheme is not \eff, the
largest payment is strictly less than $1$ larger than the other
payments with probability more than zero. We will use this to
construct an instance over which there is constant probability (over
the history) that the myopic agent chooses an unfair distribution over
arms. Additional complications arise from the fact that the principal
learns about the instance from the set of previous unfair rounds
(which, in the classic case, we did not have, since we only argued
there had to be a single unfair round if the payment scheme was not
\eff). We circumvent this problem by arguing that the principal must
deploy a \eff distribution to be fair, even if the principal knows
everything about the instance $I$ and even if the principal knows the
empirical estimates $\hat{\theta}^t_i$ for all $t\in [T], i \in [k]$.

\begin{proof}
  Consider the one-dimensional case, where $\theta_i \in \R_{\geq 0}$.
  We construct an instance $I$ such that even for a principal who has full
  information about $I$, and 
  $\hat{\theta}^t_i$ for all $t' \leq t,i$, in order to guarantee that the payment distribution in round $t$ is fair for any set of arriving contexts $x^t$,
  the largest payment must be at least $1-\eta$ with probability
  $\frac{k-1}{k}$. This clearly holds for any round in which a \eff
  payment distribution is used, and so for the remainder, we assume that the
  distribution in round $t$ is not \eff.

  Let $\theta_i = 1-\eta \in (0, 1)$ for all $i$, and let arm $1$
  have deterministic rewards equal to their mean, so that
  $\theta_1 x^t_1 = x^t_1$ for all $t$. Because the rewards are deterministic and the agent is using an ordinary least squares estimator, the myopic agent's
  prediction $\hat{\theta}^t_1 = \theta^t_1$ as well for all
  $t$. For all
  $i\neq 1$, let
  $\D^t_{i, x^t_i} = U[\theta_i x^t_i - \epsilon, \theta_i x^t_i +
  \epsilon]$
  for some very small $\epsilon$. 
  $\Pr{\rew{t}{i}\sim \D^t_{i, x^t_i}}{\rew{t}{i} > \theta_i x^t_i} =
  \frac{1}{2} = \Pr{\rew{t}{i}\sim \D^t_{i, x^t_i}}{\rew{t}{i} <
    \theta_i x^t_i} $:
  the rewards drawn from these distributions have the right
  expectation but are always larger or smaller than their expectation,
  and each with equal probability. Then, again by properties of the ordinary least squares estimator, this will imply that with
  probability $\frac{1}{2}$ over observations, in any round $t$ and
  for any $i \in[k]\setminus \{1\}$,
  $\hat{\theta}^t_{i} > \theta^t_{i}$, and with probability
  $\frac{1}{2}$, $\hat{\theta}^t_{i}< \theta^t_{i}$, for any round
  $t$. Furthermore, with probability $1$, in every round $t$, every
  empirical estimate of the coefficients is distinct:
  $\hat{\theta}^t_i \neq \hat{\theta}^t_j$ for all $i \neq j \in [k]$.

  We begin by arguing that every coordinate $i$ must have equal
  probability of receiving the largest payment in any round $t$ if the
  round is to be fair (with probability $1-\delta$ over the history). Precisely, fix some history $h^t$,
  and let
\[d_i = \Pr{p\sim \payp^t(h^t)}{i\textrm{ wins with payment vector }p , x^t_i = \vec{0},  \forall i| h^t}.\]
Since for all $i\in [k]$ and any $h^t$, $\theta_i \cdot x^t_i = 0$,
it must be that $d_i = \frac{1}{k}$ for all $i$ if round $t$
is fair for this $h^t$.

We have assumed the payment scheme is not \eff in round $t$,
conditioned on some particular history $h^t$. Thus,
\[\Pr{p\sim \payp^t(h^t)}{\max_i p_i - \max_{j\neq i} p_j \geq 1} < 1.\]
We argue that round $t$ must be unfair conditioned on $h^t$ or that
with probability $\frac{k-1}{k}$, $\max_i p_i \geq 1-\eta$. Let
\[c = \sup_c : \Pr{p\sim \payp^t(h^t)}{\max_i p_i - \max_{j\neq i} p_j \geq c} = 1;\]
\newcommand{\bigi}{{\overline{i}}}
\newcommand{\li}{{\underline{i}}}
\noindent we again know that some such $c\geq 0$ must exist, and that
$c < 1$ because the payment scheme is un\eff. Let arm $\bigi$ have the
largest empirical coefficient:
$\hat{\theta}^t_\bigi > \max_{i'\neq \bigi} \hat{\theta}^t_{i'}$ in
round $t$ and arm $\li$ have the smallest empirical coefficient,
$\hat{\theta}^t_\li < \min_{i'\neq \li}\hat{\theta}^t_{i'}$.  Further
define
\[c_{i'} = \sup_{c} : \Pr{p\sim \payp^t(\cdot)}{ p_{i'} - p_\bigi \geq
  c | p_{i'} \geq \max_{i''\neq i'} p_{i''}} = 1,\]
e.g. that $c_{i'}$ is the margin by which $i'$ has payment larger than
arm $\bigi$ when $i'$ has largest payment. Note $c_{i'} \in [0,1]$ for
all $i'$.  Let $i_{\max} \in \argmax_{i'} c_{i'}$ be an arm with
largest payment margin over $\bigi$ and
$i_{\min} = \argmin_{i'} c_{i'}$ be the arm with the smallest payment
margin over $i$.  We consider three cases: when
$c_{i_{\max}} > 1-\eta$, when
$1 - \eta \geq c_{i_{\max}} > c_{i_{\min}}$, and when
$1 - \eta \geq c_{i_{\max}} = c_{i_{\min}}$.  In each case, we show
that either the largest payment is at least $1-\eta$ with probability
at least $\frac{k-1}{k}$, or the round is unfair.

\paragraph{Case 1: $c_{i_{\max}}> 1- \eta$}

We claim here that either $c_{i_{\min}} > 1-\eta$ or the round is
unfair: this will imply that with probability $\frac{k-1}{k}$,
$\max_i p_i \geq 1-\eta$. Suppose the round is fair. Consider the
context $x^t_\bigi = \frac{1-\eta}{\hat{\theta}^t_\bigi}$ and
$x^t_{i'} = 0$ for all $i' \neq \bigi$. Then,
$\hat{\theta}^t_\bigi x^t_\bigi = 1-\eta$, and
$\hat{\theta}^t_i x^t_i = 0 = \theta_i x^t_i $ for all other
$i$. Fairness will imply that all $i \neq 1$ should be played with
equal probability. Notice that $i_{\max}$ is played with probability
$\frac{1}{k}$: precisely when $i_{\max}$ has the largest payment
(which must be largest by $c_{i_{\max}}> 1-\eta$).  $i_{\min}$ wins
only when her payment is largest (which happens with probability
$\frac{1}{k}$) and larger than $\bigi$'s by at least $1-\eta$. So, if
$\ppc{t}{i_{\min}}{h^t} = \ppc{t}{i_{\max}}{h^t} = \frac{1}{k}$, it
must be that $c_{i_{\min}} \geq 1-\eta$.

\paragraph{Case 2: $1 - \eta \geq c_{i_{\max}} >  c_{i_{\min}}$}
We argue that the round must be unfair if $c_{i_{\max}} > c_{i_{\min}}$.

Choose contexts $x^t_i$ such that
$\hat{\theta}^t_\bigi x^t_\bigi = c_{i_{\max}} \leq 1-\eta$ and
$x^t_{i'} = 0$ for all other $i'$. Then, since $\theta_i x^t_i = 0$
for all $i\neq \bigi$, if this round is to be fair, all arms
$i \neq \bigi$ must be played with equal probability. Arm $i_{\max}$
wins whenever it has the largest payment, since
$p_{i_{\max}} \geq p_\bigi + c_{i_{\max}}$ whenever $i_{\max}$ has
the largest payment. Therefore $i_{\max}$ wins with probability
$\frac{1}{k}$.

$i_{\min}$, on the other hand, wins only when they have the largest
payment and beat $i$'s payment by $c_{i_{\max}} > c_{i_{\min}}$, which
happens with strictly less probability than $i_{\min}$ having largest
payment (probability $\frac{1}{k}$) by the definition of
$c_{i_{\min}}$. So, $i_{\min}$ wins with probability strictly less
than that of $i_{\max}$; this round must be unfair.

\paragraph{Case 3: $1 - \eta \geq c_{i_{\max}} = c_{i_{\min}}$}

In this case, $c_a = c_b = \beta$ for all
$a, b \in [k]\setminus \{\bigi\}$. If $\beta \geq 1-\eta$, the claim
holds (the largest payment is at least $1-\eta$ with probability at
least $\frac{k-1}{k}$, so assume $\beta < 1-\eta$.

Suppose $\beta > 0$. We exhibit a set of contexts for which this
payment scheme combined with the agent is unfair. Fix some
$D \in (\beta, 1-\eta)$; define the contexts
\begin{itemize}
\item
$x^t_\bigi : \hat{\theta}^t_\bigi x^t_\bigi = D > \beta$
\item
$x^t_{j} : \hat{\theta}^t_j x^t_j = D-\beta > 0$ for $j$ the arm with
second-largest empirical coefficient,
\item  $x^t_{i'} = x^t_j$ for all
$i' \notin \{\bigi,j\}$.
\end{itemize}

Then, $\hat{\theta}^t_{i'}x^t_{i'} < \hat{\theta}^t_j x^t_j$, and so
arm $j$ is played whenever $j$ has the largest payment, since $j$ (and
all other arms) has margin over $\bigi$ of at least $\beta$ when they
have the largest payment; thus $j$ is played with probability
$\frac{1}{k}$. Since $\theta_j x^t_j = \theta_{i'} x^t_{i'}$ for all
$i'\neq \bigi$, if this round is fair, each $i'$ must also be played
with probability $\frac{1}{k}$, in particular for $i'$ with smallest
$\hat{\theta}^t_{i'}$. However,
$\hat{\theta}^t_{i'} x^t_{i'} < D - \beta$; $i'$ can only win if her
payment is the largest and it beats the payment of $i$ by strictly
more than $\beta$, which happens with probability strictly less than
$\frac{1}{k}$ by definition of $\beta$. Thus $i'$ cannot win with
probability as large as $j$ and so round $t$ is unfair if
$c_a = c_b = \beta > 0$ for all $a, b \neq i$.

Finally, we consider the case where $\beta = 0$ and separately argue
that this round cannot be fair. The contexts $x^t_{i'}= 1$ for all
$i'$ should prove this: arm $\bigi$ will be played with probability
$\frac{1}{k}$ (precisely the probability that $\bigi$ gets the weakly
largest payment), but arms with smaller empirical means will need to
have the largest payment by some margin, which happens with strictly
less probability than them having the largest payment by the
definition of $\beta$, so they win with probability less than
$\frac{1}{k}$, meaning fairness is violated in this round, since
$\theta_\bigi x^t_\bigi = 1-\eta = \theta_{i'} x^t_{i'}$.
\end{proof}

\section{Full Information: Perfect Fairness with Sublinear Payments}\label{sec:full-info}

In this section, we show that a principal with full information about
the state of a myopic agent can design a payment scheme which is fair in
every round and has sublinear cost for both the classic and linear
contextual bandits problems. This contrasts with the partial
information model, where for $k\geq 3$ arms, in both the classic and
linear contextual settings, in which there must be unfair rounds for any payment
scheme with total cost $o(T)$.

Roughly, the fair payment scheme operates as follows. In each round,
the scheme knows the empirical estimates of rewards used by the myopic
agent. Moreover, the scheme can compute confidence intervals around these
estimates (the scheme knows how many times each arm was pulled, and,
in a contextual setting, the contexts for each previous choice). In
such a round, the payment scheme then will choose an arm $i$ uniformly
at random from the set of arms chained to the arm with highest upper
confidence bound, and offer a payment for choosing $i$ equal to the
difference between the empirical estimate of $i$'s reward and the
empirical estimate of the highest reward in that round. This induces
uniformly random play amongst the top set of arms, and by
Lemma~\ref{lem:intervals}, this will be a fair distribution.

We now present the pseudocode in Figure~\ref{fig:full-info} a parametrized family of payment schemes described informally above. A payment scheme in this family is instantiated by giving a method of constructing valid confidence intervals around myopic predictions. 

\newcommand{\fairalg}{\textrm{Fair-Payments}}
\begin{algorithm2e}[h!]
\Fn{\fairalg($\delta, T$)}{
$\as \gets \{1,\ldots,k\}$\;
          \While{$t \leq T$} {
            $i^t= \argmax_i \hat{\mu}^t_{i^t}$\;
           Let $\as^t = \{i \textrm{ chained to } i^t \textrm{ by $\delta$-valid confidence intervals from round $t$} \}$\;
           Choose $j^t \in_{\uar} \as$\; \tcp*{Pick an arm in the upper chain to incentivize}
           Offer $p^t$ : $p^t_{j^t} = \hat{\mu}^t_{i^t} - \hat{\mu}^t_{j^t}, p^t_{i' \neq j^t} = 0$\;
         }
}
\caption{A Fair Full Information Payment Scheme}\label{fig:full-info}
\end{algorithm2e}

\begin{theorem}\label{thm:full-upper}
  Consider an instance of $\fairalg(\delta, T)$ instantiated with confidence intervals
  $[\low{t}{i}, \up{t}{i}]$ such that
  $\hat{\mu}^t_i = \frac{\up{t}{i} - \low{t}{i}}{2}$, and with
  probability $1-\delta$, for all $i\in [k], t\in [T]$,
  $\mu^t_i \in [\low{t}{i}, \up{t}{i}]$. Then, $\fairalg(\delta, T)$
  is fair at every round, and has cost and regret
  $O(k \sum_t w(t) + \delta T)$, where $w(t)$ is the maximum width of
  any confidence interval in the top chained set.
\end{theorem}
Before proving Theorem~\ref{thm:full-upper}, we mention that this theorem, when
combined with standard methods of constructing confidence intervals, implies the existence of fair payment
schemes with sublinear cost and regret, both in the classic and linear
contextual settings.

\begin{corollary}\label{cor:classic}
  Consider the classic bandits problem. Then, $\fairalg(\delta, T)$
  using the confidence interval for arm $i$ introduced by
  $\CW(\delta, T, n^t_i)$ is fair and has cost and regret
  $O(\sqrt{k^3 T\ln\frac{kT}{\delta}})$.
\end{corollary}

\begin{proof}
  By Lemma~\ref{lem:width}, with probability $1-\delta$, these
  confidence intervals are all valid for all $t\in [T], i \in [k]$.
  So, Theorem~\ref{thm:full-upper} applies, and states that this
  payment scheme is fair, and has regret $O(k \cdot \sum_t w(t))$,
  where $w(t)$ is the maximum width of any arm in the active set at
  round $t$. Since the chained set is monotone, at round $t$ every arm
  in the chained set has been chained for $t$ rounds. Therefore, in
  expectation each arm in the chain has been pulled $\frac{t}{k}$
  times. An additive Chernoff bound implies that any particular arm
  has, with probability at least $1-\frac{\delta}{2kt^2}$, been pulled
  in round $t$ at least
  $\frac{t}{k} - \sqrt{\frac{t \ln\left(\frac{2t^2k}{\delta}
      \right)}{2}}$
  times, and so this bound holds for all rounds and all arms with
  probability at lest $1-\frac{\delta}{2}$ summing up over all $k$
  arms and all $t$. Then, by Lemma 3 in~\citet{joseph2016nips}, we
  know that
  $w(t) \leq 2 \sqrt{\frac{\ln\left((\pi
        t)^2/3\delta\right)}{2\frac{t}{k} -
      \sqrt{\frac{t\ln\left(\frac{2kt^2}{\delta}\right)}{2}}}}$. Summing over all $t$ we have the desired result.
\end{proof}

\begin{corollary}\label{cor:linear}
  Consider the linear contextual bandits problem. Suppose the myopic agent uses a  ridge regression estimator:
  $\hat{\theta}^t_i = \left(X^T_i X_i + \lambda I\right)^{-1}X^T_i
  Y_i$,
  where $X_i, Y_i$ are the design matrices and observations before
  round $t$. Then, define
  \[w^t_i = ||x^t_i||_{\left(X^T_i X_i + \lambda
      I\right)^{-1}}(m\sqrt{d\ln\frac{1 + t/\lambda}{\delta}} +
  \sqrt{\lambda},\]
  and
  \[\low{t}{i} = \langle \hat{\theta}^t_i , x^t_i \rangle - w^t_i ,\hspace{.2in}
  \up{t}{i} = \langle \hat{\theta}^t_i , x^t_i \rangle + w^t_i.\]
  Then, $\fairalg(\delta, T)$ is fair and has cost and regret
  $ O\left(m d \sqrt{k^3 T} \ln^2\frac{T^2k}{d \lambda
      \delta}\right).$
\end{corollary}
\begin{proof}
  The analysis in the proof of Theorem 2 of~\citet{joseph2016rawlsian}
  shows that these confidence intervals are valid with probability
  $1-\delta$.  Their analysis upper-bounds $k \sum_t w(t)$ where
  $w(t)$ is the largest width of any confidence interval in the
  chained set in round $t$, by
  \[O\left(m d \sqrt{k^3 T} \ln^2\frac{T^2k}{d \lambda \delta}\right).\]
  Thus, the resulting algorithm is fair, and the bounds on cost
  and regret follow from Theorem~\ref{thm:full-upper}.
\end{proof}

We now proceed with the proof of Theorem~\ref{thm:full-upper}.

\begin{proof}
  We begin by proving that $\fairalg(\delta, T)$ is fair in every
  round.  By assumption, with probability $1-\delta$, for all
  $i\in [k], t\in [T]$, $\mu^t_i \in [\low{t}{i}, \up{t}{i}]$.  Thus
  it follows from Lemma~\ref{lem:intervals} that it suffices to
  show that these payments suffice to induce uniformly random play
  amongst the set of arms chained to the arm with upper confidence
  bound. By definition, the top chain in round $t$ is exactly $\as^t$.
  The distribution over payments in round $t$ chooses each
  $j^t\in \as^t$ with probability $\frac{1}{|\as^t|}$ and accordingly
  a $p^t$ such that
  $p^t_{j^t} = \hat{\mu}^t_{i^t} - \hat{\mu}^t_{j^t}$ and
  $p^t_{i'} = 0$. This induces the myopic agent to choose $j^t$ in all
  such cases. Thus, each $j^t\in \as^t$ is chosen by the myopic agent
  with probability $\frac{1}{|\as^t|}$. So, $\fairalg(\delta, T)$ is
  fair.

  Condition on all confidence intervals being valid.  The myopic agent
  under this payment scheme chooses uniformly at random from the top
  chain, which has regret in round $t$ bounded by
  $\sum_{i\in \as^t}\up{t}{i} - \low{t}{i} \leq k w(t)$, where
  $w(t) = \max_{i\in \as^t}\up{t}{i} - \low{t}{i}$.  Thus, in total,
  the regret is upper bounded by $k \sum_t w(t)$.  Moreover, the
  payment in round $t$ is
  $\hat{\mu}^t_{i^t} - \hat{\mu}_{j^t}^t \leq \up{t}{i^t} -
  \low{t}{j^t} \leq \sum_{i\in \as^t}\up{t}{i} - \low{t}{i} \leq k
  w(t)$,
  and so the same bound holds for the cost of the payment scheme. With
  probability $\delta$, the widths of the confidence intervals could
  be arbitrary, as could the inaccuracy of the sampled means. An
  additive $\delta T$ bounds the additional expected regret.
\end{proof} 

\section{Conclusion and Open Questions}\label{sec:conclusion} 
Our interest in this paper is the information that a principal needs to have about the environment before he can cost-effectively incentivize a short-sighted agent to behave ``fairly.'' We focus on two information models: the partial information model---when the principal can only observe the decisions the agent made, but not their rewards (or, in the contextual case, the contexts informing those actions), and the full information model, where the principal observes everything the agent does. In the full information model, it is possible to have it all---the principal can with sub-linear total cost incentivize the agent to play fairly at every round, and obtain no regret. In the partial information model, things are more difficult. However, despite showing the impossibility of non-trivially guaranteeing fairness at every round in the classic setting, we show that with sub-linear payments, the principal can incentivize that all but a constant number of rounds are fair, and that the agent obtains a no-regret guarantee. In the linear contextual bandit setting, our results in the partial information model are strongly negative---it is not possible to obtain a sub-linear number of unfair rounds with sub-linear payments. There are many open questions, but here we mention two that we find particularly interesting:
\begin{enumerate}
\item Our bounds (both upper and lower) in the linear contextual bandit setting are for \emph{adversarially selected contexts}. In the natural case in which contexts are drawn from an (unknown) probability distribution, it may still be possible to obtain positive results in the partial information setting, analogous to the results we obtain for the classic bandits problem. However, our upper bound technique from the classic case does not directly extend to the linear contextual case even when there is a distribution over contexts. 
\item The friction to fairness here is that the agent in question has a short horizon for which he is optimizing. We study the extreme case in which he is entirely myopic. How do our results extend in the case in which the agent is not completely myopic, but is instead optimizing with respect to some fixed discount factor bounded away from 1? 
\end{enumerate}

% Bibliography
\bibliographystyle{plainnat} 
\bibliography{sources}

\begin{thebibliography}{14}
\providecommand{\natexlab}[1]{#1}
\providecommand{\url}[1]{\texttt{#1}}
\expandafter\ifx\csname urlstyle\endcsname\relax
  \providecommand{\doi}[1]{doi: #1}\else
  \providecommand{\doi}{doi: \begingroup \urlstyle{rm}\Url}\fi

\bibitem[Barocas and Selbst(2016)]{barocas2016big}
Solon Barocas and Andrew~D Selbst.
\newblock Big data's disparate impact.
\newblock \emph{CALIFORNIA LAW REVIEW}, 104:\penalty0 671, 2016.

\bibitem[Brill(2015)]{bigdata2}
FTC~Commisioner Brill.
\newblock Navigating the ``trackless ocean'': Fairness in big data research and
  decision making.
\newblock Keynote Address at the Columbia University Data Science Institute,
  April 2015.

\bibitem[Che and Horner(2015)]{che2015optimal}
Yeon-Koo Che and Johannes Horner.
\newblock Optimal design for social learning.
\newblock Technical report, Cowles Foundation for Research in Economics, Yale
  University, 2015.

\bibitem[Coglianese and Lehr(2016)]{cary16}
Cary Coglianese and David Lehr.
\newblock Regulating by robot: Administrative decision-making in the
  machine-learning era.
\newblock \emph{Georgetown Law Journal}, Forthcoming:\penalty0 Forthcoming,
  2016.

\bibitem[Edelman et~al.(2017)Edelman, Luca, and Svirsky]{edelman16}
Benjamin Edelman, Michael Luca, and Dan Svirsky.
\newblock Racial discrimination in the sharing economy: Evidence from a field
  experiment.
\newblock \emph{American Economic Journal: Applied Economics}, 9\penalty0
  (2):\penalty0 1--22, 2017.

\bibitem[Frazier et~al.(2014)Frazier, Kempe, Kleinberg, and
  Kleinberg]{frazier2014incentivizing}
Peter Frazier, David Kempe, Jon Kleinberg, and Robert Kleinberg.
\newblock Incentivizing exploration.
\newblock In \emph{Proceedings of the fifteenth ACM conference on Economics and
  computation}, pages 5--22, New York, NY, 2014. ACM, ACM.

\bibitem[Joseph et~al.(2016{\natexlab{a}})Joseph, Kearns, Morgenstern, Neel,
  and Roth]{joseph2016rawlsian}
Matthew Joseph, Michael Kearns, Jamie Morgenstern, Seth Neel, and Aaron Roth.
\newblock Rawlsian fairness for machine learning.
\newblock Technical report, University of Pennsylvania, 2016{\natexlab{a}}.

\bibitem[Joseph et~al.(2016{\natexlab{b}})Joseph, Kearns, Morgenstern, and
  Roth]{joseph2016nips}
Matthew Joseph, Michael Kearns, Jamie~H Morgenstern, and Aaron Roth.
\newblock Fairness in learning: Classic and contextual bandits.
\newblock In D.~D. Lee, M.~Sugiyama, U.~V. Luxburg, I.~Guyon, and R.~Garnett,
  editors, \emph{Advances in Neural Information Processing Systems 29}, pages
  325--333. Curran Associates, Inc., Red Hook, NY, 2016{\natexlab{b}}.

\bibitem[Kremer et~al.(2014)Kremer, Mansour, and Perry]{kremer2014implementing}
Ilan Kremer, Yishay Mansour, and Motty Perry.
\newblock Implementing the ``wisdom of the crowd''.
\newblock \emph{Journal of Political Economy}, 122\penalty0 (5):\penalty0
  988--1012, 2014.

\bibitem[Mansour et~al.(2015)Mansour, Slivkins, and
  Syrgkanis]{mansour2015bayesian}
Yishay Mansour, Aleksandrs Slivkins, and Vasilis Syrgkanis.
\newblock Bayesian incentive-compatible bandit exploration.
\newblock In \emph{Proceedings of the Sixteenth ACM Conference on Economics and
  Computation}, pages 565--582, New York, NY, 2015. ACM, ACM.

\bibitem[Mansour et~al.(2016)Mansour, Slivkins, Syrgkanis, and Wu]{mansour16}
Yishay Mansour, Aleksandrs Slivkins, Vasilis Syrgkanis, and Zhiwei~Steven Wu.
\newblock Bayesian exploration: Incentivizing exploration in bayesian games.
\newblock In \emph{Proceedings of the 2016 {ACM} Conference on Economics and
  Computation, {EC} '16, Maastricht, The Netherlands, July 24-28, 2016}, page
  661, 2016.
\newblock \doi{10.1145/2940716.2940755}.

\bibitem[O'Leary(2013)]{leary}
Daniel O'Leary.
\newblock Exploiting big data from mobile device sensor-based apps: Challenges
  and benefits.
\newblock \emph{MIS Quarterly Executive}, 12\penalty0 (4):\penalty0 179--187,
  December 2013.

\bibitem[Papanastasiou et~al.(2017)Papanastasiou, Bimpikis, and
  Sava]{papanbimp}
Yiangos Papanastasiou, Kostas Bimpikis, and Nicos Sava.
\newblock Crowdsourcing exploration.
\newblock \emph{Management Science}, 2017.
\newblock Forthcoming.

\bibitem[Pope and Snyder(2011)]{pope}
Devin Pope and Justin Snyder.
\newblock What's in a picture? evidence of discrimination from prosper.com.
\newblock \emph{Journal of Human Resources}, 101\penalty0 (1):\penalty0 53--92,
  2011.

\end{thebibliography}
\appendix 
\section{Missing Proofs}\label{sec:missing}

\begin{proof}[Proof of Theorem~\ref{thm:lower-partial-twoarms}]
  First, we will show that the payment scheme will incentivize the
  agent to select the arms fairly over all rounds with probability
  $(1 - \delta)$ over the realization of the history. By
  Lemma~\ref{lem:width}, we know that with probability at least
  $1-\delta$ over the realizations of the rewards, for all rounds $t$
  and both arms $i$,
  \begin{equation}\label{eq:dude}
    |\hat \mu_i^t - \mu_i| < \frac{\CW(\delta, t, n_i^t)}{2}.
  \end{equation}
  We will condition on this event for the remainder of the argument.
  Note that in each round $t$, there are two cases. In the first case,
  the empirical mean rewards of the two arms satisfy
  $|\hat \mu_1^t - \hat \mu_2^t | < p(\delta, t, n_1^t, n_2^t)$. Then
  the arm will be selected uniformly at random, so the algorithm is
  fair.

  In the second case, the empirical mean rewards satisfy
  $|\hat \mu_1^t - \hat \mu_2^t | \geq p(\delta, t, n_1^t,
  n_2^t)$. Without loss of generality, let us assume
  $\hat \mu_1^t > \hat \mu_2^t$, so the algorithm will
  deterministically always play arm 1. Then it follows
  from~\eqref{eq:dude} and the definition of
  $p(\delta, t, n_1^t, n_2^t)$ that the true mean rewards
  $\mu_1 > \mu_2$. Therefore, the algorithm is fair in this case (and
  also in all future rounds).

  Next, we will bound the total expected payment made by the
  principal. As a first step, we can show that with probability at
  least $(1 - 1/T)$ that, no arm is played for more than $2\log T$
  consecutive times.  Then we can bound the total payment as follows:
  \begin{align*}
\mathbb{E}\left[\sum_{t=1}^T p(\delta, t, n_1^t , n_2^t)\right] & = \mathbb{E}\left[\sum_{t=1}^T \left(                                              \CW(\delta, t, n_1^t) + \CW(\delta, t, n_2^t) \right) \right]\\ 
                                             & \leq (1 - 1/T) \left(4\log(T) \sum_{t=1}^T  \CW(\delta, t, t) \right) + (1/T) T  \, O(\log(T/\delta))\\
                                             &= O\left(\sqrt{T} \ln^2(T/\delta) \right)
 \end{align*}

 Finally, we will bound the expected cumulative regret incurred by the
 algorithm. Without loss of generality, we will assume $\mu_1 > \mu_2$
 and let $\Delta = \mu_1 - \mu_2$ be the difference between the mean
 expected rewards. The algorithm incurs an expected regret of $\Delta$
 whenever it plays the arm 2. It suffices to bound the number of times
 arm 2 is played.

 Note that the algorithm will stop playing arm 2 if at some round $t$,
\[
  \frac{\CW(\delta, t, n_1^t)}{2}, \frac{\CW(\delta, t, n_2^t)}{2}
  \leq \Delta/3
\]
This will require the algorithm to play both arms
$S= O\left(\frac{\ln(T/\delta)}{\Delta^2} \right)$ number of times. By
applying Chernoff bounds, we know that with probability at least
$(1 - 1/T)$, it suffices to have $t\geq O(S)$. This will allow us to
upper bound the expected regret by
\[
  (1 - 1/T) \Delta O\left( \frac{\ln(T/\delta)}{\Delta^2}\right) +
  T(1/T) = O\left(\frac{\ln(T/\delta)}{\Delta}\right)
\]
Also note that the expected regret is also trivially upper bounded by
$\Delta T$, so the expected regret is no more than
\[
  \min\{O\left(\frac{\ln(T/\delta)}{\Delta}, \Delta T\right)\} \leq
  O\left(\sqrt{T \ln(T/\delta)} \right),
\]
where the inequality follows from the fact that
$\min\{a, b\}\leq \sqrt{ab}$ for any $a, b>0$.
\end{proof}

\begin{proof}[Proof of Observation~\ref{obs:confidence}]
  By Lemma~\ref{lem:intervals}, the width of the confidence intervals
  produce by \CW have this property in all rounds. Since $x$ is
  defined to be the either an ``old'' value of the output of $\CW$ (in
  which case the validity of the definition of $\CW$'s confidence
  widths gives us this guarantee), or the largest output of $\CW$ for
  some $i\in \as^t$, this property continues to hold.
\end{proof}

\begin{proof}[Proof of Lemma~\ref{lem:findchain-upper}]
  We prove this iteratively over all inputs and outputs of
  \findchain. The input to the first call to \findchain is
  $\as = [k]$, so this is trivially true. Now, suppose the upper chain
  is included in the input $\as$ to \findchain: we will argue that it
  continues to be included in the output $R = \findchain(x, \as, t)$.
  We actually prove something stronger: any arm in the input $\as$
  will be in $R$ if its empirical mean in the output round $t'$ is
  within $2x$ of anything in $R$'s empirical mean in round $t'$, and
  that the arm with highest upper confidence bound in round $t'$ is in
  $R$. These two together imply $R$ contains the upper chain. Let
  $R^t = \emptyset, R^{t+\ell} =$ the set $R$ in round $t+\ell$ before
  $R$ is output by \findchain, so $R^{t'} = R$.

  Note that $R^{t+1} =\{i_0\}$ where $i_0$ is the arm with highest
  empirical mean in $\as$ with highest empirical mean at round
  $t$. Then, either $R^{t+2} = \{i_0\}$ and no arm has empirical mean
  within $2x$ of $i_0$ in round $t+1$ or $R^{t+2} = \{i_0, i_1\}$ for
  $i_1 \in \as$ such that
  $\hat{\mu}^{t+1}_{i_1} \geq \hat{\mu}^{t+1}_{i_0} - 2 x$.  More
  generally, in round $t + \ell < t'$, an arm was added in round
  $t + \ell -1$, and another arm will be added in round $t + \ell$ if
  if some empirical mean in \as is within $2x$ of any arm already
  belonging to $R = \{i_0, i_1, \ldots, i_{t+\ell - 1}\}$, since the
  payments for any arm in $\as$ but not $R$ increases by $2x$ in this
  round. Thus, every arm which is not in the output $R$ must be more
  than $2x$ away from any arm in the output $R$ in round $t'$.

  We argue now that the arm with highest upper confidence bound in
  round $t'$ belongs to $R$.  Since $\as$ contained the upper chain in
  round $t$ by assumption, it in particular contains the arm with the
  highest upper confidence in round $t'$. Thus, it suffices to argue
  that $R$ contains the arm in $\as$ with highest upper confidence
  bound in round $t'$. Either arm $i_0$ or some arm with
  empirical mean within $2x$ of $i_0$ at round $t'$ must be the arm
  with highest upper confidence bound in round $t'$ amongst those arms
  in $\as$, since $x$ is an upper bound on the width of the confidence
  intervals of the set of arms in $\as$, and by the previous argument,
  $R$ contains every arm whose empirical mean in round $t'$ is within
  $2x$ of $i_0$'s mean.
\end{proof}

\begin{proof}[Proof of Lemma~\ref{lem:findchain-chained}]
  We first prove the first claim.  This is true for the first input to
  \findchain: we argue that for any input \as which contains every arm
  chained to an arm in \as, $R = \findchain(x,\as,t)$ contains any arm
  chained to an arm in $R$. So, suppose $\as$ contains every arm
  chained to an arm in \as.  By the argument used in the proof of
  Lemma~\ref{lem:findchain-upper}, any arm with an empirical mean
  within $2x$ of any arm in $R$'s empirical mean in round $t'$ belongs
  to $R$.  So, any arm linked to an arm in $R$ must belong to $R$ and
  therefore so must any arm chained to $R$.

  Now, we argue that for any arm $i \notin R$ must have
  $\up{t}{i} < \min_{i' \in R} \low{t}{i'}$. This is true for the
  first input to \findchain (the entire set $[k]$ is the initial
  input). We argue that conditioned on this holding for a particular
  input to \findchain, the output from \findchain will also satisfy
  this claim.  Notice that every arm in $i' \in R$ was incentivized to
  be played during this call to \findchain, and those arms no longer
  in $R$ were not, which means their empirical means were more than
  $2x$ away from any arm ultimately in $R$; that for $i\notin R$,
  $\hat{\mu}^t_i + 2x < \min_{i'\in R}\hat{\mu}^t_{i'}$.  Thus, since
  $\low{t}{j} = \hat{\mu}^t_j - x$ and $\up{t}{j} = \hat{\mu}^t_j + x$
  for all $j\in \as$, the claim follows.
\end{proof}

\begin{proof}[Proof of Lemma~\ref{lem:findchain-problem}]
  The empirical mean of any arm in the output of \findchain had to be
  within $2x$ of some arm in $R$ when it was input. Those empirical
  means might change (each mean by at most $x$, since the confidence
  intervals are valid, by Observation~\ref{obs:confidence}), so the
  empirical differences might change by $2 \cdot x$, but the total
  difference between these two arms is then increased by at most
  $2 x$. Thus, summing up these distances gives the second claim.
\end{proof}

\begin{proof}[Proof of Lemma~\ref{lem:chainfair-uniform}]
  Offering payment of $4 x |\as|$ for arm $i'$ and payment $0$ for all
  other arms $j$ in round $t$ will cause a myopic agent to choose $i$
  if
  $\max_{i \in \{k\}}\hat{\mu}^t_{i} - \hat{\mu}^t_{i'} \leq 4 |\as|
  \cdot x$.
  Thus, if $\argmax_{i \in \{k\}}\hat{\mu}^t_{i} \in \as$, each
  $i'\in \as$ will be played with equal probability by construction of
  the payment vectors used in
  \chainfair. Lemma~\ref{lem:findchain-upper} implies the top chain
  and therefore the top arm are contained in \as, thus, the claim
  holds.
\end{proof}

\begin{proof}[Proof of Lemma~\ref{lem:chainfair-find}]
  With probability $\frac{1}{|\as|} \geq \frac{1}{k}$, arm $i^t$
  selected by \chainfair for payment $4 x \cdot |\as|$. Since
  $\max_{j \in \as}|\hat{\mu}^t_{j} - \hat{\mu}^t_i| > 4 |\as| \cdot
  x$,
  the myopic agent will prefer the arm with highest mean, and his
  choice will therefore cause \findchain to be called. The probability that such an arm $i^t$ is not called for $k\log(1/\delta)$ consecutive rounds is at most $\delta$. 
\end{proof}

\end{document}